\documentclass[twoside,11pt]{article}

\usepackage{jmlr2e}

\usepackage{graphicx} 
\usepackage{subfigure} 
\usepackage{color}
\usepackage{algorithm,algorithmic}
\usepackage{amsmath}
\usepackage{array}
\usepackage{verbatim}
\usepackage{graphics}
\usepackage{epstopdf}
\usepackage{multicol}

\usepackage{natbib}

\usepackage{algorithm}
\usepackage{algorithmic}

\usepackage{hyperref}

\jmlrheading{?}{2012}{??-??}{6/12}{??/??}{Reza Bosagh Zadeh and Ashish Goel}

\ShortHeadings{Dimension Independent Similarity Computation}{Bosagh Zadeh and Goel}
\firstpageno{1}

\begin{document}

\title{Dimension Independent Similarity Computation}

\author{\name Reza Bosagh Zadeh \email rezab@stanford.edu \\
       \addr Institute for Computational and Mathematical Engineering\\
       Stanford University\\
       Stanford, CA 94305, USA
       \AND
       \name Ashish Goel \email ashishg@stanford.edu \\
       \addr Department of Management Science and Engineering\\
       Stanford University\\
       Stanford, CA 94305, USA}
       
\editor{Inderjit Dhillon}

\maketitle

\begin{abstract}
We present a suite of algorithms for Dimension Independent Similarity Computation (DISCO)
to compute all pairwise similarities between very high-dimensional sparse vectors.
All of our results are provably independent of dimension, meaning that
apart from the initial cost of trivially reading in the data, all subsequent
operations are independent of the dimension; thus the dimension can be very large.
We study Cosine, Dice, Overlap, and the Jaccard similarity measures. 
For Jaccard similiarity we include an improved version of MinHash.
Our results are geared toward the MapReduce framework. We empirically validate our
theorems with large scale experiments using data from the social networking site Twitter. 
At time of writing, our
algorithms are live in production at twitter.com.
\end{abstract}


\section{Introduction}

Computing similarity between all pairs of vectors in large-scale datasets is a challenge. Traditional approaches
of sampling the dataset are limited and linearly dependent on the dimension of the data. We present an approach whose complexity is independent of the data dimension and geared towards modern distributed systems, in particular the MapReduce framework \citep{mapreduce}. 

MapReduce is a programming model for processing large data sets, typically used to do distributed computing on clusters of commodity computers. With large amount of processing power at hand, it is very tempting to solve problems by brute force. However, we show how to combine clever sampling techniques with the power of MapReduce to extend its utility.

Consider the problem of finding all pairs of similarities between $D$ indicator (0/1 entries) vectors, each of dimension $N$. In particular we focus on cosine similarities between all pairs of $D$ vectors in $\mathbb{R}^N$. Further assume that each dimension is $L$-sparse, meaning each dimension has at most $L$ non-zeros across all points. For example, typical values to compute similarities between all pairs of a subset of Twitter users can be: $N = 10^9$ (the universe of Twitter users), $D = 10^7$ (a subset of Twitter users), $L = 1000$. To be more specific: each of $D$ users can be represented by an $N$ dimensional vector that indicates which of the $N$ universal users are followed. The sparsity parameter ($L=1000$) here assumes that each user follows at most 1000 other users.

There are two main complexity measures for MapReduce: ``shuffle size" and ``reduce-key complexity", defined shortly \citep{ashishperf}. It can be easily shown that the naive approach for computing similarities will have $O(NL^2)$ emissions, which for the example parameters we gave is infeasible. The number of emissions in the map phase is called the ``shuffle size", since that data needs to be shuffled around the network to reach the correct reducer. Furthermore, the maximum number of items reduced to a single key can be as large as $N$. Thus the ``reduce-key complexity" for the naive scheme is $N$.

We can drastically reduce the shuffle size and reduce-key complexity by some clever sampling with the DISCO scheme
described in this paper. In this case, the output of the reducers are random variables whose expectations are the similarities. Two proofs are needed to justify the effectiveness of this scheme: first, that the expectations are indeed correct and obtained with high probability, and second, that the shuffle size is greatly reduced. We prove both of these claims. In particular, in addition to correctness, we prove that to estimate similarities above $\epsilon$, the shuffle size of the DISCO scheme is only $O(DL \log(D)/\epsilon)$, with no dependence on the dimension $N$, hence the name.

This means as long as there are enough mappers to read the data, one can use the DISCO sampling scheme to make the shuffle size tractable. Furthermore, each reduce key gets at most $O(\log(D)/\epsilon)$ values, thus making the reduce-key complexity tractable too. Within Twitter inc, we use the DISCO sampling scheme to compute similar users. We have also used the scheme to find highly similar pairs of words by taking each dimension to be the indicator vector that signals in which tweets the word appears. We empirically verify the proven claims and observe large reductions in shuffle size in this paper.

Our sampling scheme can be used to implement many other similarity measures.
We focus on four similarity measures: Cosine, Dice, Overlap, and the Jaccard similarity measures. 
For Jaccard similiarity we present an improved version of the well known MinHash scheme \citep{minhash}. Our framework operates
under the following assumptions, each of which we justify:

First, we focus on the case where each dimension is sparse and therefore the natural way to store the data is by segmenting it into dimensions. In our application each dimension is represented by a tweet, thus this assumption was natural.
Second, our sampling scheme requires a ``background model", meaning the magnitude of each vector is assumed to be known
and loaded into memory. In our application this was not a hurdle, since the magnitudes of vectors in our corpus needed to be
computed for other tasks. In general this may require another pass over the dataset. To further address the issue, in the streaming computation model, we can remove the dependence by paying an extra logarithmic factor in memory used.
Third, we prove results on highly similar pairs, since common applications require thresholding the similarity score with a high threshold value.

A ubiquitous problem is finding all pairs of objects that are in some sense similar and in particular more similar than a threshold. For such applications of similarity, DISCO is particularly helpful since higher similarity pairs are estimated with provably better accuracy. There are many examples, including

\begin{itemize}
\item Advertiser keyword suggestions: When targeting advertisements via keywords, it is useful to expand the manually input 
set of keywords by other similar keywords, requiring finding all keywords more similar than a high threshold \citep{regelson2006predicting}. The vector representing a 
keyword will often be an indicator vector indicating in which dimensions the keywords appear.
\item Collaborative filtering: Collaborative filtering applications require knowing which users have similar interests. For this reason, given a person or user, it is required to find all other objects more similar than a particular threshold, for which our algorithms are well suited.
\item Web Search: Rewriting queries given to a search engine is a common trick used to expand the coverage of the search results \citep{abhishek2007keyword}. Adding clusters of similar queries to a given query is a natural query expansion strategy.
\end{itemize}

These applications have been around for many years, but the scale at which we need to solve them keeps steadily 
increasing. This is particularly true with the rise of the web, and social networks such as Twitter, Facebook, and Google Plus among others. Google currently holds an index of more than 1 trillion webpages, on which duplicate detection must be a daunting task. For our experiments we use the Twitter dataset, where the number of tweets in a single day surpasses 200 million. Many of the applications,
including ours involve domains where the dimensionality of the data far exceeds the number of points. In our case the dimensionality is large (more than 200 million), and we can prove that apart from initially reading the data (which cannot be avoided), 
subsequent MapReduce operations will have input size independent of the number of dimensions. Therefore, as long as the data can
be read in, we can ignore the number of dimensions.

A common technique to deal with large datasets is to simply sample. Indeed, our approach too, is a sampling scheme; 
however, we sample in a nontrivial way so that points that have nonzero
entries in many dimensions will be sampled with lower probability than points that are only present in a few dimensions.
Using this idea, we can remove the dependence on dimensionality while being able to mathematically prove -- and empirically verify -- accuracy.

Although we use the MapReduce \citep{mapreduce} framework and discuss shuffle size, the sampling strategy we use can be
generalized to other frameworks. We focus on MapReduce because it is the tool of choice for large scale
computations at Twitter, Google, and many other companies dealing with large scale data. 
However, the DISCO sampling strategy is potentially useful whenever one is computing a number between 0 and 1
by taking the ratio of an unknown number (the dot product in our case) by some known number (e.g. the magnitude). This is a high-level description, and we give four concrete examples, along with proofs, and experiments. Furthermore, DISCO improves the implementation of the well-known MinHash \citep{minhash} scheme in any MapReduce setup.

\section{Formal Preliminaries} \label{sec:formal}

Let $T = \{t_1, \ldots, t_N\}$ represent $N$ documents, each of a length of no more than $L$ words. In our context, the documents are tweets (documents are also the same as dimensions for us) from the social networking site Twitter, or could be fixed-window contexts from a large corpus of documents, or any other dataset where dimensions have no more than $L$ nonzero entries and the dataset is available dimension-by-dimension. We are interested in similarity scores between pairs of words in a dictionary containing $D$ words $\{w_1, \ldots, w_D\}$. The number of documents in which two words $w_i$ and $w_j$ co-occur is denoted $\#(w_i, w_j)$. The number of documents in which a single word $w_i$ occurs is denoted $\#(w_i)$. 

To each word in the dictionary, an $N$-dimensional indicator vector is assigned, indicating in which documents the word occurs. We operate within a MapReduce framework where each document is an input record. We denote the number of items output by the map phase as the `shuffle size'. Our algorithms will have shuffle sizes that provably do not depend on $N$, making them particularly attractive for very high dimensional tasks. 

We focus on 5 different similarity measures, including cosine similarity, which is very popular and produces high quality results across different domains \citep{chien2005semantic,chuang2005taxonomy,sahami2006web,spertus2005evaluating}. Cosine similarity is simply the vector normalized dot product:$\frac{\#(x, y)}{\sqrt{\#(x)} \sqrt{\#(y)}}$ where $\#(x) = \sum_{i=1}^N x[i]$ and $\#(x,y) = \sum_{i=1}^N x[i]  y[i]$.
In addition to cosine similarity, we consider many variations of similarity scores that use the dot product. They are outlined in Table \ref{measures}.

To compare the performance of algorithms in a MapReduce framework, we report and analyze shuffle size, which is more reliable than wall time or any other implementation-specific measure. We define shuffle size as the total output of the Map Phase, which is what will need to be ``shuffled" before the Reduce phase can begin. Our results and theorems hold across any MapReduce implementation such as Hadoop \citep{hadoop}\citep{gates2009building} or Google's MapReduce \citep{mapreduce}. We focus on shuffle size because after one trivially reads in the data via mappers, the shuffle phase is the bottleneck since our mappers and reducers are all linear in their input size. In general, MapReduce algorithms are usually judged by two performance measures: largest reduce bucket and shuffle size. Since the largest reduce bucket is not at all a problem for us, we focus on shuffle size.

For many applications, and our applications in particular, input vectors are sparse, meaning that the large majority of entries are 0. A sparse vector representation for a vector $x$ is the set of all pairs $(i,x[i])$ such that $x[i]>0$ over all $i = 1\ldots N$. The size of a vector $x$, which we denote as $\#(x)$ is the number of such pairs. We focus on the case where each \textit{dimension} is sparse and therefore the natural way to store the data is segmented into dimensions.

Throughout the paper we formally prove results for pairs more similar than a threshold, called $\epsilon$. Our algorithms
will often have a tradeoff between accuracy and shuffle size, where the tradeoff parameter is $p$. By design, the larger $p$ becomes,
the more accurate the estimates and the larger the shuffle size becomes.

We assume that the dictionary can fit into memory but that the number of dimensions (documents) is so large that many machines will be needed to even hold the documents on disk. Note that documents and dimensions are the same thing, and we will use these terms interchangeably. We also assume that the magnitudes of the dimensions are known and available in all mappers and reducers.
\newcommand\T{\rule{0pt}{3.0ex}}
\newcommand\B{\rule[-1.2ex]{0pt}{0pt}}

\begin{table}[H]          
\begin{center}
\setlength{\columnsep}{1.5in}

  \begin{tabular}{ | c | c | c | c |}
    \hline
    \textbf{Similarity} &  \textbf{Definition} & \textbf{Shuffle Size} & \textbf{Reduce-key size}\\  \hline    
    Cosine \T &  $\frac{\#(x, y)}{\sqrt{\#(x)} \sqrt{\#(y)}}$ & $O(D L \log(D)/\epsilon) $ & $O(\log(D)/\epsilon) $ \\[8pt]  \hline
    Jaccard \T  & $\frac{\#(x, y)}{ \#(x) + \#(y) - \#(x, y)}$ & $O((D/\epsilon) \log(D/\epsilon))$ & $O(\log(D/\epsilon)/\epsilon) $ \\[8pt] \hline
    Overlap \T & $\frac{\#(x, y)}{\min(\#(x), \#(y))}$ & $O(D L \log(D)/\epsilon) $ & $O(\log(D)/\epsilon) $ \\[8pt] \hline
    Dice \T & $\frac{2 \#(x, y)}{ \#(x) + \#(y) }$ & $O(D L \log(D)/\epsilon) $ & $O(\log(D)/\epsilon) $ \\[8pt] 
    \hline 
  \end{tabular}
   \caption{Similarity measures and the bounds we obtain. All  sizes are independent of $N$, the dimension. \label{measures}These are bounds for shuffle size without combining. Combining can only bring down these sizes.}   
\end{center}
\end{table}
We are interested in several similarity measures outlined in Table \ref{measures}. For each, we prove the shuffle size for the bottleneck step of the pipeline needed to compute all nonzero scores. Our goal is to compute similarity scores between all pairs of words, and prove accuracy results for pairs which have similarity above $\epsilon$. The naive approach is to first compute the dot product between all pairs of words in a Map-Reduce \citep{mapreduce} style framework. Mappers act on each document, and emit key-value pairs of the form $(w_i, w_j) \rightarrow 1$. These pairs are collected with the sum reducer, which gives the dot product $\#(w_i, w_j)$. The product is then used to trivially compute the formulas in Table \ref{measures}. The difficulty of computing the similarity scores lies almost entirely in computing $\#(w_i, w_j)$.

Shuffle size is defined as the number of key-value pairs emitted. The shuffle size for the above naive approach is $N {L \choose 2} = O(NL^2)$, which can be prohibitive for $N$ too large. Our algorithms remove the dependence on $N$ by dampening the number of times popular words are emitted as part of a pair (e.g. see algorithm \ref{alg:sampleemit}). Instead of emitting every pair, only some pairs are output, and instead of computing intermediary dot products, we directly compute the similarity score. 

\section{Complexity Measures}

An in-depth discussion of complexity measures for MapReduce is given in \cite{ashishperf}. We focus on shuffle size
and reduce-key complexity for two reasons. First,
focusing on reduce key complexity captures the performance of an idealized Map-Reduce system with
infinitely many mappers and reducers and no coordination overheads. A small reduce-key complexity guarantees that a Map-Reduce algorithm will not suffer from ``the curse of the last reducer" \citep{lastreducer}, a phenomenon in which the average work done by all reducers may be small, but due to variation in the size
of reduce records, the total wall clock time may be extremely large or, even worse, some reducers
may run out of memory.

Second, we focus on shuffle size since the total size of input/output
by all mappers/reducers captures the total file I/O done by the algorithm, and is often of the
order of the shuffle size.

Lastly, our complexity measures depend only on the algorithm, not on details of the Map-Reduce installation
such as the number of machines, the number of Mappers/Reducers etc., which is a desirable property
for the analysis of algorithms \citep{ashishperf}.

\section{Related Work} \label{sec:related}

Previously the all-pairs similarity search problem has been studied in \citep{broder1997syntactic, minhash}, in the context of identifying near-duplicate web pages \citep{broder1997syntactic}. We describe MinHash later in section \ref{sec:minhash}. We improve the 
vanilla implementation of MinHash on MapReduce to arrive at better shuffle sizes without (effectively) any loss in accuracy. We prove these
results in Section \ref{sec:minhash}, and verify them experimentally.

There are many papers discussing the all pairs similarity computation problem. In \citep{elsayed2008pairwise,rasmus}, the MapReduce framework is targeted, 
but there is still a dependence on the dimensionality. In \citep{lin2009brute}, all pairs are computed on MapReduce, but there is a focus on the life sciences domain. The all-pairs similarity search problem has also been addressed in the database community, where it is known as the similarity join problem \citep{arasu2006efficient, chaudhuri2006primitive, sarawagi2004efficient}. These papers are loosely assigned to one of two categories: First, signatures of the points to convert the nonexact matching problem to an exact matching problem followed by hash map and filtering false positives. Second, inverted list based solutions exploit information retrieval techniques.

In \cite{baraglia}, the authors consider the problem of finding highly similar pairs of documents in MapReduce. They discuss shuffle size briefly but do not provide proven guarantees. They also discuss reduce key size, but again do not provide bounds on the size. 

There is large body of work on the nearest neighbors problem, which is the problem of finding the $k$ nearest neighbors
of a given query point\citep{charikar2002similarity, fagin2003efficient, gionis1999similarity, indyk1998approximate}. Our problem, the all-pairs similarities computation is related to, but not identical to the nearest neighbor problem. As one of the
reviewers astutely pointed out, when only the $k$ nearest neighbors need to be found,
  optimizations can be made for that.  In particular, there are two
  problems with pairwise similarity computation, large computation
  time with respect to dimensions which we address, and large
  computation time with respect to the number of points.

In \citep{pantel2009web}, the authors propose a highly scalable term similarity algorithm, implemented in the MapReduce framework, and deployed over a 200 billion word crawl of the Web to compute pairwise similarities between terms. Their results are still dependent upon the dimension and they only provide experimental evaluations of their algorithms.

Other related work includes clustering of web data \citep{beeferman2000agglomerative, chien2005semantic, sahami2006web, spertus2005evaluating}. These applications of clustering typically employ relatively straightforward exact algorithms or approximation methods for computing similarity. Our work could be leveraged by these applications for improved performance or higher accuracy through reliance on proven results. We have also used DISCO in finding similar users in a production environment \cite{wtfpaper}.

\section{Algorithms and Methods} \label{sec:results}
The Naive algorithm for computing similarities first computes
dot products, then simply divides by whatever is necessary to obtain a similarity score.
i.e. in a MapReduce implementation:
\begin{enumerate}
\item Given document $t$, Map using NaiveMapper (Algorithm \ref{alg:stupidmapper})
\item Reduce using the NaiveReducer (Algorithm \ref{alg:stupidreducer})
\end{enumerate}

\algsetup{indent=2em}
\begin{algorithm}[H]
\caption{NaiveMapper$(t)$}\label{alg:stupidmapper}
\begin{algorithmic} [4]
\FOR{all pairs $(w_1, w_2)$ in $t$}
	\STATE emit $((w_1, w_2) \rightarrow 1)$
\ENDFOR
\end{algorithmic}
\end{algorithm}

\algsetup{indent=2em}
\begin{algorithm}[H]
\caption{NaiveReducer$((w_1, w_2), \langle r_1, \ldots, r_R \rangle)$}\label{alg:stupidreducer}
\begin{algorithmic} [4]
\STATE $a =  \sum_{i=1}^{R} r_i$
\STATE output $\frac{a}{\sqrt{\#(w_1)\#(w_2)}}$
\end{algorithmic}
\end{algorithm}

The above steps will compute all dot products, which will then be scaled by appropriate factors for each of the similarity scores. Instead of using naive algorithm, we modify the mapper of the naive algorithm and replace it with Algorithm \ref{alg:discomapper}, and replace the reducer with Algorithm \ref{alg:discoreducer} to directly compute the actual similarity score, not dot products.
The following sections detail how to obtain dimensionality independence for each of the similarity scores.

\algsetup{indent=2em}
\begin{algorithm}[H]
\caption{DISCOMapper$(t)$}\label{alg:discomapper}
\begin{algorithmic} [4]
\FOR{all pairs $(w_1, w_2)$ in $t$}
	\STATE emit using custom Emit function
\ENDFOR
\end{algorithmic}
\end{algorithm}

\algsetup{indent=2em}
\begin{algorithm}[H]
\caption{DISCOReducer$((w_1, w_2), \langle r_1, \ldots, r_R \rangle)$}\label{alg:discoreducer}
\begin{algorithmic} [4]
\STATE $a =  \sum_{i=1}^{R} r_i$
\STATE output $a \frac{\epsilon}{p}$
\end{algorithmic}
\end{algorithm}

\subsection{Cosine Similarity}

To remove the dependence on $N$, we replace the emit function with Algorithm \ref{alg:sampleemit}.

\algsetup{indent=2em}
\begin{algorithm}[H]
\caption{CosineSampleEmit$(w_1, w_2)$ - $p/\epsilon$ is the oversampling parameter}\label{alg:sampleemit}
\begin{algorithmic} [4]
\STATE With probability $$\frac{p}{\epsilon} \frac{1}{\sqrt{\#(w_1)} \sqrt{\#(w_2)}}$$ emit $((w_1, w_2) \rightarrow 1)$
\end{algorithmic}
\end{algorithm}

Note that the more popular a word is, the less likely it is to be output. This is the key observation leading to shuffle size independent of the dimension. We use a slightly different reducer, which instead of calculating $\#(w_1, w_2)$, computes $\cos(w_1, w_2)$ directly without any intermediate steps. The exact estimator is given in the proof of Theorem \ref{expgood}.

Since the CosineSampleEmit algorithm is only guaranteed to produce the correct similarity score in expectation, we must show that the expected value is highly likely to be obtained. This guarantee is given in Theorem \ref{expgood}. 

\begin{theorem} \label{expgood}
For any two words $x$ and $y$ having $\cos(x,y) \geq \epsilon$, let $X_1, X_2, \ldots, X_{\#(x,y)}$ represent indicators for the coin flip in calls to CosineSampleEmit with $x, y$ parameters, and let $X = \sum_{i=1}^{\#(x,y)} X_i$. For any $1> \delta > 0$, we have
$$\Pr \left[ \frac{\epsilon}{p} X > (1+\delta)\cos(x,y) \right] \leq \left(\frac{e^\delta}{(1+\delta)^{(1+\delta)}}\right)^p $$
and 
$$\Pr \left[ \frac{\epsilon}{p} X < (1-\delta)\cos(x,y) \right]  < \exp(-p\delta^2/2)$$
\end{theorem}
\begin{proof}
We use $\frac{\epsilon}{p} X$ as the estimator for $\cos(x,y)$. Note that 
$$\mu_{xy} = E [X] = \#(x,y) \frac{p}{\epsilon} \frac{1}{\sqrt{\#(x)} \sqrt{\#(y)}} = \frac{p}{\epsilon} \cos(x,y) \geq  p$$ 
Thus by the multiplicative form of the Chernoff bound,
$$\Pr \left[ \frac{\epsilon}{p} X > (1+\delta)\cos(x,y) \right]
=\Pr \left[ \frac{\epsilon}{p} X > (1+\delta)\frac{\epsilon}{p}  \frac{p}{\epsilon} \cos(x,y) \right]$$
$$=\Pr \left[ X > (1+\delta)\mu_{xy}\right] < \left(\frac{e^\delta}{(1+\delta)^{(1+\delta)}}\right)^{\mu_{xy}}
\leq \left(\frac{e^\delta}{(1+\delta)^{(1+\delta)}}\right)^p$$

Similarly, by the other side of the multiplicative Chernoff bound, we have
$$\Pr \left[ \frac{\epsilon}{p} X < (1-\delta)\cos(x,y) \right] =\Pr[X < (1-\delta)\mu_{xy}] < \exp(-\mu_{xy}\delta^2/2) \leq \exp(-p\delta^2/2)$$
\end{proof}

Since there are ${D \choose 2}$ pairs of words in the dictionary, set $p = \log(D^2) = 2 \log(D)$
and use union bound with theorem \ref{expgood} to ensure the above bounds hold simultaneously for all pairs $x,y$ having $\cos(x,y) \geq \epsilon$. It is worth mentioning that $\left(\frac{e^\delta}{(1+\delta)^{(1+\delta)}}\right)$ and $\exp(-\delta^2/2)$ are always
less than 1, thus raising them to the power of $p$ brings them down exponentially in $p$. Note that although we decouple $p$ and
$\epsilon$ for the analysis of shuffle size, in a good implementation that are tightly coupled and the expression $p/\epsilon$
can be thought of as a single parameter to tradeoff accuracy and shuffle size. In a large scale experiment described in section \ref{largeexperiment} we set $p/\epsilon$ in a coupled manner.

Now we show that the shuffle size is independent of $N$, which is a great improvement over the naive approach when $N$ is large.
To see the usefulness of these bounds, it is worth noting that we prove they are almost optimal (i.e. no other algorithm can do much better). In Theorem \ref{smallshuffle} we prove that any algorithm that purports to accurately calculate highly similar pairs must at least output them, and sometimes there are at least $DL$ such pairs, and so any algorithm that is accurate on highly similar pairs must have at least $DL$ shuffle size. We are off optimal here by only a $\log(D)/\epsilon$ factor.

\begin{theorem} \label{smallshuffle}
The expected shuffle size for CosineSampleEmit is $O(D L \log(D) / \epsilon)$ and $\Omega(DL)$.
\end{theorem}
\begin{proof}
The expected contribution from each pair of words will constitute the shuffle size:
$$\sum_{i=1}^{D} \sum_{j=i+1}^D \sum_{k=1}^{\#(w_i, w_j)} \text{Pr}[\text{CosineSampleEmit}(w_i, w_j)] $$
$$= \sum_{i=1}^{D} \sum_{j=i+1}^D \#(w_i, w_j) \text{Pr}[\text{CosineSampleEmit}(w_i, w_j)]  $$
$$= \sum_{i=1}^{D} \sum_{j=i+1}^D \frac{p}{\epsilon} \frac{\#(w_i, w_j)}{\sqrt{\#(w_i)} \sqrt{\#(w_j)}} $$
$$ \leq \frac{p}{2 \epsilon} \sum_{i=1}^{D} \sum_{j=i+1}^D \#(w_i, w_j)( \frac{1}{\#(w_i)} + \frac{1}{\#(w_j)}) $$
$$\leq \frac{p}{\epsilon} \sum_{i=1}^{D} \frac{1}{\#(w_i)} \sum_{j=1}^D \#(w_i, w_j)$$
$$\leq \frac{p}{\epsilon} \sum_{i=1}^{D} \frac{1}{\#(w_i)} L \#(w_i) = \frac{p}{\epsilon} LD = O(DL \log(D)/\epsilon)$$

The first inequality holds because of the Arithmetic-Mean-Geometric-Mean inequality applied to $\{1/\#(w_i), 1/\#(w_j)\}$. The
last inequality holds because $w_i$ can co-occur with at most $\#(w_i) L$ other words. 
It is easy to see via Chernoff bounds that the above shuffle size is obtained with high probability.

To see the lower bound, we construct a dataset consisting of $D/L$ distinct documents of length $L$,
furthermore each document is duplicated $L$ times. To construct this dataset, consider 
grouping the dictionary into $D/L$ groups, each group containing $L$ words. A document is associated
with every group, consisting of all the words in the group. This document is then repeated $L$ times.
In each group, it is trivial to check that all pairs of words of have similarity exactly 1. There are ${L \choose 2}$
pairs for each group and there are $D/L$ groups, making for a total of
$(D/L) {L \choose 2} = \Omega(DL)$ pairs with similarity 1, and thus also at least $\epsilon$.
Since any algorithm that purports to accurately calculate highly-similar pairs must at least \textit{output} them,
and there are $\Omega(DL)$ such pairs, we have the lower bound.
\end{proof}

It is important to observe what happens if the output `probability' is greater than 1. We certainly Emit, but when the output probability is greater than 1, care must be taken during reducing to scale by the correct factor, since it won't be correct to divide by $p/ \epsilon$, which is
the usual case when the output probability is less than 1. Instead, we must divide by $\sqrt{\#(w_1)} \sqrt{\#(w_2)}$ because
for the pairs where the output probabilty is greater than 1, CosineSampleEmit and Emit are the same. Similar corrections have to be
made for the other similarity scores (Dice and Overlap, but not MinHash), so we do not repeat this point. Nonetheless it is an important one which arises during implementation.

Finally it is easy to see that the largest reduce key will have at most
$p / \epsilon  = O(\log(D)/\epsilon)$ values.
\begin{theorem} \label{smallreducekey}
The expected number of values mapped to a single key  by DISCOMapper is $p / \epsilon$.
\end{theorem}
\begin{proof}
Note that the output of DISCOReducer is a number between 0 and 1. Since this is obtained
by normalizing the sum of all values reduced to the key by $p / \epsilon$, and all summands are at most
1, we trivially get that the number of summands is at most $p / \epsilon$.
\end{proof}

\subsection{Jaccard Similarity} \label{sec:minhash}
Traditionally MinHash \citep{minhash} is used to compute Jaccard similarity scores between all pairs in a dictionary.
We improve the MinHash scheme to run much more efficiently with a smaller shuffle size.

Let $h(t)$ be a hash function that maps documents to distinct numbers in $[0,1]$, and for any word $w$ define $g(w)$ (called the MinHash of $w$) to be the minimum value of $h(t)$ over all $t$ that contain $w$. Then $g(w_1)=g(w_2)$ exactly when the minimum hash value of the union with size $\#(w_1) + \#(w_2) - \#(w_1, w_2)$ lies in the intersection with size $\#(w_1, w_2)$. Thus
$$\text{Pr}[g(w_1)=g(w_2)] = \frac{\#(w_1, w_2)}{\#(w_1) + \#(w_2) - \#(w_1, w_2)} = \text{Jac}(w_1, w_2)$$

Therefore the indicator random variable that is 1 when $g(w_1)=g(w_2)$ has expectation equal to the Jaccard similarity between the two words. Unfortunately it has too high a variance to be useful on its own. The idea of the MinHash scheme is to reduce the variance by averaging together $k$ of these variables constructed in the same way with $k$ different hash functions. We index these $k$ functions using the notation $g_j(w)$ to denote the MinHash of hash function $h_j(t)$. We denote the computation of hashes as `MinHashMap'. Specifically, MinHashMap is defined as Algorithm \ref{alg:minhashemit}.

To estimate $\text{Jac}(w_1, w_2)$ using this version of the scheme, we simply count the number of hash functions for which $g(w_1)=g(w_2)$, and divide by $k$ to get an estimate of $\text{Jac}(w_1, w_2)$. By the multiplicative Chernoff bound for sums of 0-1 random variables as seen in Theorem \ref{expgood}, setting $k = O(1/\epsilon)$ will ensure that w.h.p. a similarity score that is above $\epsilon$ has relative error no more than $\delta$. Qualitatively, this theorem is the same as given in \citep{minhash} (where MinHash is introduced) and we do not claim the following as new contribution, however, we include it for completeness. More rigorously, 
\begin{theorem} \label{minhashgood}
Fix any two words $x$ and $y$ having $\text{Jac}(x,y) \geq \epsilon$. Let $X_1, X_2, \ldots, X_k$ represent indicators for $\{g_1(x)=g_1(y),\ldots, g_k(x)=g_k(y)\}$ and $X = \sum_{i=1}^{k} X_i$. For any $1> \delta > 0$ and $k = c/\epsilon$, we have
$$\Pr \left[  X/k > (1+\delta)\text{Jac}(x,y) \right] \leq \left(\frac{e^\delta}{(1+\delta)^{(1+\delta)}}\right)^c$$ and
$$\Pr \left[ X/k < (1-\delta)\text{Jac}(x,y) \right] \leq \exp(-c\delta^2/2)$$
\end{theorem}
\begin{proof}
We use $X/k$ as the estimator for $\text{Jac}(x,y)$. Note that $E [X] = k \text{Jac}(x,y) = (c/\epsilon) \text{Jac}(x,y) \geq  c$. Now by standard Chernoff bounds we have 
$$\Pr \left[  X/k > (1+\delta)\text{Jac}(x,y) \right] = \Pr \left[  X > (1+\delta)E[X] \right] 
 \leq \left(\frac{e^\delta}{(1+\delta)^{(1+\delta)}}\right)^{E[X]} \leq \left(\frac{e^\delta}{(1+\delta)^{(1+\delta)}}\right)^c$$
Similarly, by the other side of the multiplicative Chernoff bound, we have
$$\Pr \left[ X/k < (1-\delta)\text{Jac}(x,y) \right] \leq \exp(-c\delta^2/2)$$
\end{proof}

The MapReduce implementation of the above scheme takes in documents and for each unique word in the document outputs $k$ hash values. 

\algsetup{indent=2em}
\begin{algorithm}[H]
\caption{MinHashMap$(t, \langle w_1, \ldots, w_L \rangle)$}\label{alg:minhashemit}
\begin{algorithmic} 
\FOR{$i=1$ to $L$}
	\FOR{$j=1$ to $k$}
		\STATE emit $((w_i, j) \rightarrow h_j(t))$
	\ENDFOR
\ENDFOR
\end{algorithmic}
\end{algorithm}

\begin{enumerate}
\item Given document $t$, Map using MinHashMap (Algorithm \ref{alg:minhashemit})
\item Reduce using the min reducer
\item Now that  minhash values for all hash functions are available, for each pair we can simply compute the number of 
hash collisions and divide by the total number of hash functions
\end{enumerate}

\algsetup{indent=2em}
\begin{algorithm}[H]
\caption{MinHashSampleMap$(t, \langle w_1, \ldots, w_L \rangle)$}\label{alg:minhashsampleemit}
\begin{algorithmic} 
\FOR{$i=1$ to $L$}
	\FOR{$j=1$ to $k$}
		\IF{$h_j(t) \leq \frac{c \log(Dk)}{\#(w_i)}$} 
			\STATE emit $((w_i, j) \rightarrow h_j(t))$
		\ENDIF
	\ENDFOR
\ENDFOR
\end{algorithmic}
\end{algorithm}

Recall that a document has at most $L$ words. This naive Mapper will have shuffle size $N L k = O(NL/\epsilon)$, which can be improved upon. After the map phase, for each of the $k$ hash functions, the standard MinReducer is used, which will compute $g_j(w)$. These MinHash values are then simply checked for equality. We modify the initial map phase, and prove that the modification brings down shuffle size while maintaining correctness. The modification is seen in algorithm \ref{alg:minhashsampleemit}, note that $c$ is a small constant we take to be 3.

We now prove that MinHashSampleMap will with high probability Emit the minimum hash value for a given word $w$ and hash function $h$, thus ensuring the steps following MinHashSampleMap will be unaffected.

\begin{theorem} \label{jacworks}
If the hash functions $h_1, \ldots, h_k$ map documents to $[0,1]$ uniform randomly, and $c=3$, then with probability at least $1- \frac{1}{(Dk)^2}$, for all words $w$ and hash functions $h \in \{h_1, \ldots, h_k\}$, MinHashSampleMap will emit the document that realizes $g(w)$.
\end{theorem}
\begin{proof}
Fix a word $w$ and hash function $h$ and let $z = \min_{t |w \in t} h(t)$. Now the probability that
MinHashSampleMap will not emit the document that realizes $g(w)$ is 
$$ \Pr \left[ z > \frac{c \log(Dk)}{\#(w)} \right]  = \bigg(1-\frac{c \log(Dk)}{\#(w)}\bigg)^{\#(w)} \leq e^{-c \log(Dk)}  = \frac{1}{(Dk)^c}$$
Thus for a single $w$ and $h$ we have shown MinHashSampleMap will w.h.p. emit the hash that realizes the MinHash. To show the same
result for \textit{all} hash functions and words in the dictionary, 
set $c=3$ and use union bound to get a $(\frac{1}{Dk})^2$ bound on the probability of error for \textit{any} $w_1, \ldots, w_D$ and $h_1, \ldots, h_k$.  
\end{proof}

Now that we have correctness via theorems \ref{minhashgood} and \ref{jacworks}, we move onto calculating the shuffle size for MinHashSampleMap.
\begin{theorem}
MinHashSampleMap has expected shuffle size $O(D k \log(Dk)) = O((D/\epsilon) \log(D/\epsilon))$.
\end{theorem}
\begin{proof}
Simply adding up the expectations for the emissions indicator variables, we see the shuffle size is bounded by:
$$\sum_{h \in \{h_1,\ldots, h_k\}} \sum_{w \in \{w_1,\ldots, w_D\}} \frac{c \log(Dk)}{\#(w)} \#(w) = D k c \log(Dk)$$
Setting $c=3$ and $k = 1/\epsilon$ gives the desired bound.
\end{proof}

All of the reducers used in our algorithms are associative and commutative operations (sum and min), and
therefore can be combined for optimization. Our results do not change qualitatively when combiners are used, except for one case.
A subtle point arises in our claim for improving MinHash. If we combine with $m$ mappers, then the naive MinHash
implementation will have a shuffle size of $O(mDk)$ whereas DISCO provides a shuffle size of $O(Dk \log(Dk))$. Since
$m$ is usually set to be a very large constant (one can easily use 10,000 mappers in a standard Hadoop implementation), 
removing the dependence on $m$ is beneficial. For the other similarity measures, combining the Naive mappers can only bring down the shuffle size to $O(m D^2)$, which DISCO improves upon asymptotically by obtaining a bound of $O(DL/\epsilon \log(D))$ without even combining, so combining will help even more. In practice, DISCO can be easily combined, ensuring superior performance both theoretically and empirically.

\section{Cosine Similarity in Streaming Model} 
\label{sec:streaming}
We briefly depart from the MapReduce framework and instead work in the `Streaming' framework.
In this setting, data is streamed dimension-by-dimension through a single machine that can at any time answer
queries of the form ``what is the similarity between points $x$ and $y$ considering all the input so far?".
The main  performance measure is how much memory the machine uses and
queries will be answered in constant time. We describe the algorithm only for cosine similarity,
and an almost identical algorithm will work for dice, and overlap similarity.

Our algorithm will be very similar to the mapreduce setup, but in place of emitting pairs to be shuffled
for a reduce phase, we instead insert them into a hash map $H$, keyed by pairs of words, with each entry holding a bag of emissions. $H$ is used to track the emissions by storing them in a bag associated with the pair. Since all data streams through a single machine, for any
word $x$, we can keep a counter for $\#(x)$. This will take $D$ counters worth of memory, but as we will
see in Theorem 1 this memory usage will be dominated by the size of $H$. Each emission is
decorated with the probability of emission. There are two operations to be
described: the update that occurs when a new dimension (document) arrives, and the constant time algorithm
used to answer similarity queries.

\textbf{Update.} On an update we are given a document. For each pair of words $x,y$ in the document, with independent coin flips of probability $q =  \frac{p}{\epsilon} \frac{1}{ \sqrt{\#(x) \#(y)}}$ we lookup the bag associated with $(x,y)$ in $H$ and insert $q$ into it. It is important to note that
the emission is being done with probability $q$ and $q$ is computed using the \textit{current values} of $\#(x)$ and $\#(y)$. Thus if a query comes after this update, we must take into account all new information. This is done via subsampling and is explained shortly. It remains to show how
to use these probabilities to answer queries.

\textbf{Query.} We now describe how to answer the only query. Let the query be for the similarity between words $x$ and $y$.
Recall at this point we know both $\#(x)$ and $\#(y)$ exactly, for the data seen so far.
We lookup the pair $(x,y)$ in $H$, and grab the associated bag of emissions. 
Recall from above that emission is decorated with the probability of emission $q_i$ for the $i$'th 
entry in the bag. Unfortunately $q_i$ will be larger than we need it to be, since it was computed
at a previous time, when fewer occurrences of $x$ and $y$ had happened. To remedy this,
we independently subsample each of the emissions for the pair $x,y$ with coin flips of probability
$ \frac{p}{\epsilon} \frac{1}{q_i \sqrt{\#(x) \#(y)}}$. For each pair $x,y$ seen in the input, there will be
exactly 
$$q_i \frac{p}{\epsilon}  \frac{1}{q_i \sqrt{\#(x) \#(y)}} = \frac{p}{\epsilon} \frac{1}{\sqrt{\#(x) \#(y)}}$$
probability of surviving emission \textit{and} the subsampling. Finally, since the pair $x,y$ is seen
exactly $\#(x,y)$ times, the same estimator used in Theorem 2 will have 
expectation equal to $\cos(x,y)$. Furthermore, since before subsampling we output in expectation more pairs than CosineSampleEmit, 
the Chernoff bound of Theorem 2 still holds. Finally, to show that $H$ cannot grow too large, 
we bound its size in Theorem \ref{thm:streamingworks}.

\begin{theorem} \label{thm:streamingworks}
The streaming algorithm uses at most
$O(D L \lg(N)  \log(D) / \epsilon)$ memory.
\end{theorem}
\begin{proof}
We only need to bound the size of $H$.
Consider a word $x$ and all of its occurrences in documents $t_1, \ldots, t_{\#(x)}$ at final
time (i.e. after all $N$ documents have been processed).
We conceptually and only for this analysis construct
a new larger dataset $C'$ where each word $x$ is removed and its occurrences 
are replaced \textit{in order} with $\lfloor \lg\#(x) \rfloor + 1$ new words
$x_1, \ldots, x_{\lfloor \lg\#(x) \rfloor + 1}$, so we are effectively segmenting (in time) the occurrences of $x$. 
With this in mind, we construct $C'$ so that each $x_i$ will replace $2^{i-1}$ occurrences of $x$, in time order.
i.e. we will have $\#(x_i) = 2^{i-1}$.

Note that our streaming algorithm updates the counters for words with every update.
Consider what happens if instead of updating every time, the streaming algorithm somehow in advance
knew and used the final $\#(x)$ values after all documents have been processed.
We call this the `all-knowing' version. The size of $H$
for such an all-knowing algorithm is the same as the shuffle size for the DISCO sampling scheme
analyzed in Theorem 2, simply because there is a bijection between 
the emits of CosineSampleEmit and inserts into $H$ with exactly the same coin flip probability.
We now use this observation and $C'$.

We show that the memory used by $H$ when our algorithm is run on $C'$ 
with the all-knowing counters
dominates (in expectation) the size of $H$ for the original dataset in the streaming model, thus achieving
the claimed bound in the current theorem statement.

Let PrHashMapInsert($x,y$) denote the probability of inserting the pair $x,y$ when we run the all-knowing version of 
the streaming algorithm on $C'$.  
Let PrStreamEmit($x,y,a,b$) denote the probability of emitting the pair $x,y$ in
the streaming model with input $C$, after observing $x,y$ exactly $a,b$ times, respectively. With these definitions we have
$$\text{PrStreamEmit}(x,y,a,b) = \frac{p}{\epsilon} \frac{1} {\sqrt{ab}}$$
$$ \leq \frac{p}{\epsilon} \frac{1} {2^{\lfloor \lg a \rfloor} 2^{\lfloor \lg b \rfloor} }
\leq \frac{p}{\epsilon} \frac{1} {\sqrt{\#(x_{\lfloor \lg a \rfloor})} \sqrt{\#(x_{\lfloor \lg b \rfloor})} }$$
$$= \text{PrHashMapInsert}(x_{\lfloor \lg a \rfloor+1}, y_{\lfloor \lg b \rfloor+1})$$
The first inequality holds by properties of the floor function. The second inequality holds by definition of $C'$.
The dictionary size for $C'$ is O($D\lg (N)$) where $D$ is the original dictionary size for $C$. Using the same
analysis of Theorem 2, the shuffle size for $C'$
is at most $O(D L \lg(N)  \log(D) / \epsilon)$ and therefore so is the size of $H$ for the all-knowing algorithm run on $C'$,
and by the analysis above, so is the hash map size for the original dataset $C$.
\end{proof}

\section{Correctness and Shuffle Size Proofs for other Similarity Measures}
\label{sec:correct}

\subsection{Overlap Similarity}
Overlap similarity follows the same pattern as we used for cosine similarity, thus we only
explain the parts that are different. The emit function changes to Algorithm \ref{alg:overlapsampleemit}.

\algsetup{indent=2em}
\begin{algorithm}[h!]
\caption{OverlapSampleEmit$(w_1, w_2)$ - $p/\epsilon$ is the oversampling parameter} \label{alg:overlapsampleemit}
\begin{algorithmic} [4]
\STATE With probability $$\frac{p}{\epsilon} \frac{1}{\min(\#(w_1), \#(w_2))}$$ emit $((w_1, w_2) \rightarrow 1)$
\end{algorithmic}
\end{algorithm}

The correctness proof is nearly identical to cosine similarity so we do not restate it. The shuffle size for
OverlapSampleEmit is given by the following theorem.
\begin{theorem} \label{overlapsmallshuffle}
The expected shuffle size for OverlapSampleEmit is $O(D L \log(D) / \epsilon)$.
\end{theorem}
\begin{proof}
The expected contribution from each pair of words will constitute the shuffle size:
$$ \sum_{i=1}^{D} \sum_{j=i+1}^D \frac{p}{\epsilon} \frac{\#(w_i, w_j)}{\min(\#(w_i), \#(w_j))}  $$
$$\leq \frac{p}{ \epsilon} \sum_{i=1}^{D} \sum_{j=i+1}^D \#(w_i, w_j)( \frac{1}{\#(w_i)} + \frac{1}{\#(w_j)}) $$
$$\leq \frac{2p}{\epsilon} \sum_{i=1}^{D} \frac{1}{\#(w_i)} \sum_{j=1}^D \#(w_i, w_j) $$
$$\leq \frac{2p}{\epsilon} \sum_{i=1}^{D} \frac{1}{\#(w_i)} L \#(w_i) = \frac{2p}{\epsilon} LD = O(DL \log(D)/\epsilon)$$

The first inequality holds trivially. The
last inequality holds because $w_i$ can co-occur with at most $\#(w_i) L$ other words. 
It is easy to see via Chernoff bounds that the above shuffle size is obtained with high probability.
\end{proof}

\subsection{Dice Similarity}

Dice similarity follows the same pattern as we used for cosine similarity, thus we only
explain the parts that are different. The emit function changes to Algorithm \ref{alg:dicesampleemit}.

\algsetup{indent=2em}
\begin{algorithm}[h!]
\caption{DiceSampleEmit$(w_1, w_2)$  - $p/\epsilon$ is the oversampling parameter} \label{alg:dicesampleemit}
\begin{algorithmic} [4]
\STATE With probability $$\frac{p}{\epsilon} \frac{2}{\#(w_1) + \#(w_2)}$$ emit $((w_1, w_2) \rightarrow 1)$
\end{algorithmic}
\end{algorithm}

The correctness proof is nearly identical to cosine similarity so we do not restate it. The shuffle size for
DiceSampleEmit is given by the following theorem.
\begin{theorem} \label{dicesmallshuffle}
The expected shuffle size for DiceSampleEmit is $O(D L \log(D) / \epsilon)$.
\end{theorem}
\begin{proof}
The expected contribution from each pair of words will constitute the shuffle size:
$$ 2 \sum_{i=1}^{D} \sum_{j=i+1}^D \frac{p}{\epsilon} \frac{\#(w_i, w_j)}{\#(w_i) + \#(w_j)} $$
$$\leq \frac{2p}{\epsilon} \sum_{i=1}^{D} \sum_{j=1}^D  \frac{\#(w_i, w_j) }{\#(w_i)}$$
$$\leq \frac{2p}{\epsilon} \sum_{i=1}^{D} \frac{1}{\#(w_i)} L \#(w_i) = \frac{2p}{\epsilon} LD = O(DL \log(D)/\epsilon)$$

The first inequality holds trivially. The
last inequality holds because $w_i$ can co-occur with at most $\#(w_i) L$ other words. 
It is easy to see via Chernoff bounds that the above shuffle size is obtained with high probability.
\end{proof}

\section{Cross Product}

Consider the case that we have $N$ vectors (and \textit{not} $D$) in $\mathbb{R}^N$ and wish to compute the cross product of all
similarities between a subset $D$ of the $N$ vectors with the remaining $N-D$ vectors. This is different than computing
all ${D \choose 2}$ similarities which has been the focus of the previous sections. 

It is worth a brief mention that the same bounds and almost the same proof of Theorem \ref{smallshuffle}
 hold for this case, with one caveat,
that the magnitudes of all $N$ points must be available to the mappers and reducers. We can mitigate this problem
by annotating each dimension with the magnitudes of points beforehand in a separate MapReduce job. However, this MapReduce
job will have a shuffle size dependent on $N$, so we do not consider it a part of the DISCO scheme.

\section{Experiments} \label{sec:experimental}

We use data from the social networking site Twitter. Twitter is currently a very popular social network platform. Users interact with Twitter through a web interface, instant messaging clients, or sending mobile text messages. Public updates by users are viewable to the world, and a large majority of Twitter accounts are public. These public tweets provide a large real-time corpus of what is happening in the world.

We have two sets of experiments.
The first is a large scale experiment that is in production, to compute similar users in the Twitter follow graph. The second is
an experiment to find similar words. The point of the first experiment is show the scalability of our algorithms
and the point of the second set of experiments is to show accuracy.

\subsection{Similar Users}
\label{largeexperiment}
Consider the problem of finding all pairs of similarities between a subset of $D=10^7$ twitter users.
We would like to find all ${D \choose 2}$ similarities between these users. 
The number of dimensions of each user is larger than $D$, and is denoted $N=5\times 10^8$. 
We used $p/\epsilon = 100$. We have a classifier trained to pick the edges
with most interaction, and thus can limit the sparsity of the dataset with $L=1000$.

We compute these user similarities daily in a production environment using Hadoop and Pig implementations
at Twitter \citep{pig}, with the results supplementing Twitter.com. Using 100 Hadoop nodes we can compute the above similarities
in around 2 hours. In this case, the naive scheme would have shuffle size $NL^2 = 5 \times 10^{14}$, which is infeasible, however it becomes possible with DISCO.

\subsection{Similar Words}

In the second experiment we used all tweets seen by Twitter.com in one day. 
The point in the experiment is to find all pairs of similar words, where words are similar if they appear in the same tweets.
The number of dimensions $N=198,134,530$ in our data is equal to the number of tweets and each tweet is a document with size at most 140 characters, providing a small upper bound for $L$.  These documents are ideal for our framework, since our shuffle size upper bound depends on $L$, which in this case is very small. We used a dictionary of 1000 words advertisers on Twitter are currently targeting
to show Promoted Trends, Trends, and Accounts. We also tried a uniformly random sampled dictionary without a qualitative change in results.

The reason we only used $D=1000$ is because for the purpose of validating our work (i.e. reporting the small errors occurred by our algorithms), we have to compute the true cosine similarities, which means computing true co-occurence for every pair, which is a challenging task computationally. This was a bottleneck only in \textit{our experiments} for this paper, and does not affect users of our algorithms. We ran experiments with $D=10^6$, but cannot report true error since finding the true cosine similarities are too computationally intensive. In this regime however, our theorems guarantee that the results are good with high probability.

\subsection{Shuffle Size vs Accuracy} \label{sec:shufflevsaccuracy}

We have two parameters to tweak to tradeoff accuracy and shuffle size: $\epsilon$ and $p$. However,
since they only occur in our algorithms as the ratio $p/\epsilon$, we simply use that as the tradeoff parameter. The
reason we separated $p$ and $\epsilon$ was for the theorems to go through nicely, but in reality we only see a single
tweaking parameter. 

We increase $p/\epsilon$ exponentially on the x axis and record the ratio of DISCO shuffle size to the naive implementation.
In all cases we can achieve a 90\% reduction in shuffle size from the naive implementation without sacrificing much accuracy, as see in Figures \ref{fig:cosshuffle}, \ref{fig:diceshuffle}, and \ref{fig:overlapshuffle}. The accuracy we report is with respect to true cosine, dice, and overlap similarity.

\begin{figure}[h!]
  \centering
    \includegraphics[width=0.5\textwidth]{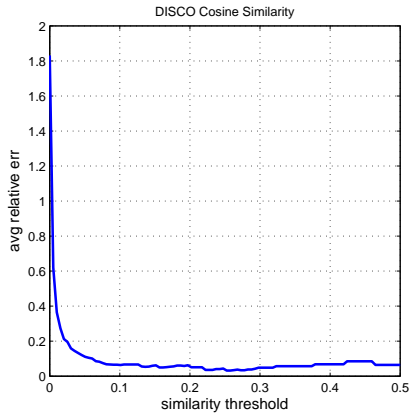}
 \caption{Average error for all pairs with similarity $\geq \epsilon$. DISCO estimated Cosine error decreases for more similar pairs.  Shuffle reduction from naive implementation: 99.39\%. \label{fig:cosallsim}}
\end{figure}
\begin{figure}[h!] 
  \centering
    \includegraphics[width=0.5\textwidth]{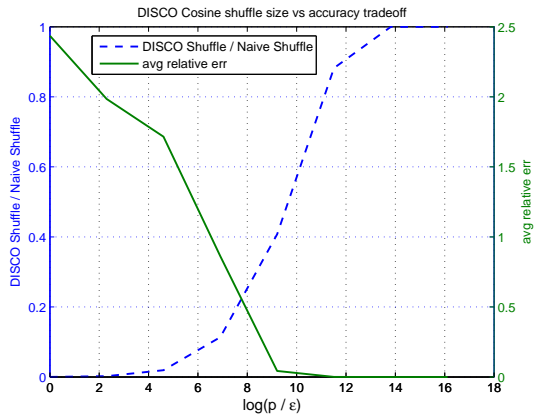}
 \caption{As $p/\epsilon$ increases, shuffle size increases and error decreases. There is no thresholding for highly similar pairs here.
 Ground truth is computed with naive algorithm. When the true similarity is zero, DISCO always also returns zero, so we always get those right. It remains to estimate those pairs with similarity $> 0$, and that is the average relative error for those pairs that we report here. \label{fig:cosshuffle}}
\end{figure}

\begin{figure}[h!]
  \centering
    \includegraphics[width=0.5\textwidth]{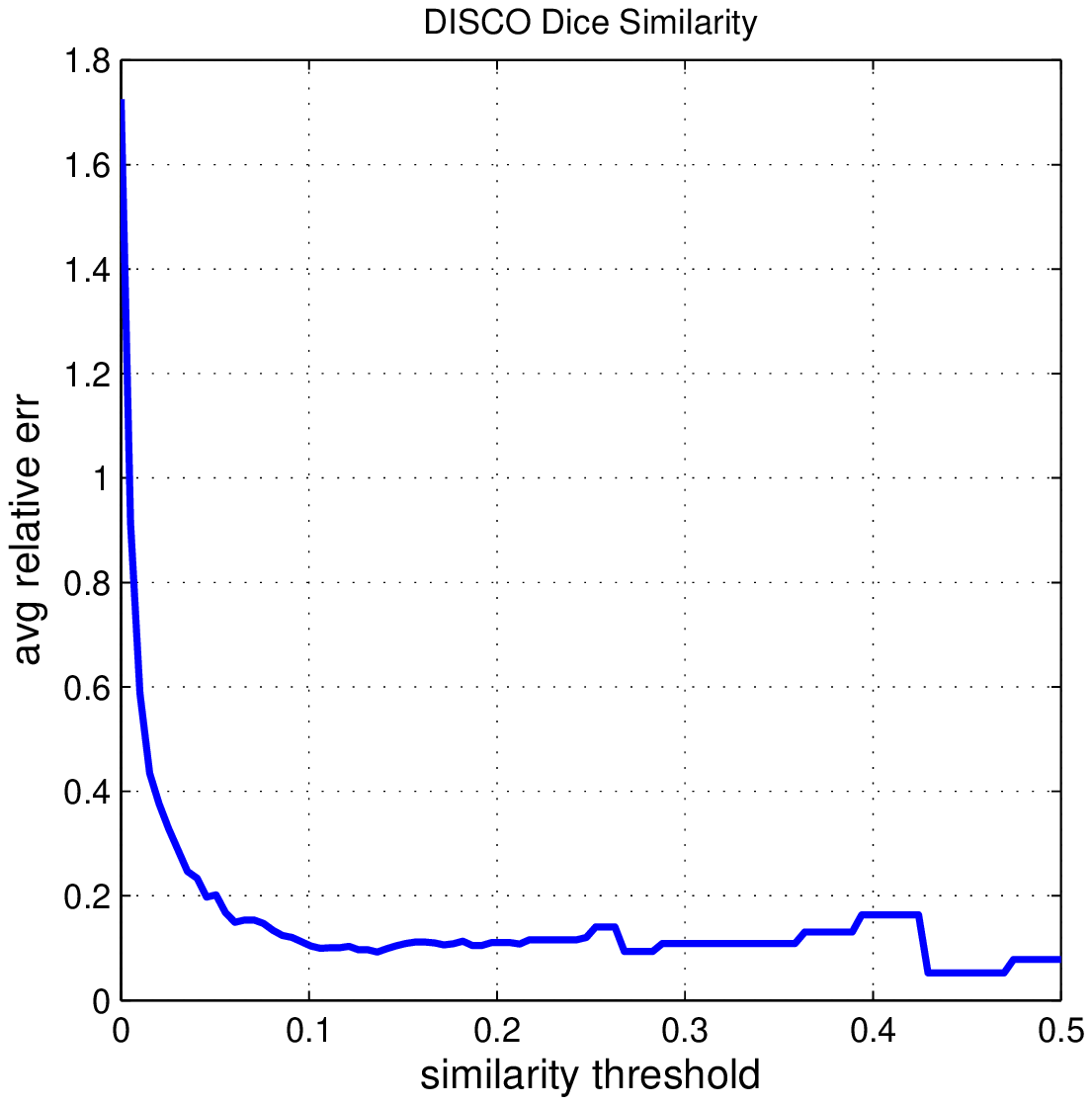}
 \caption{Average error for all pairs with similarity $\geq \epsilon$. DISCO estimated Dice error decreases for more similar pairs. Shuffle reduction from naive implementation: 99.76\%. \label{fig:diceallsim}}
\end{figure}
\begin{figure}[h!]
  \centering
    \includegraphics[width=0.5\textwidth]{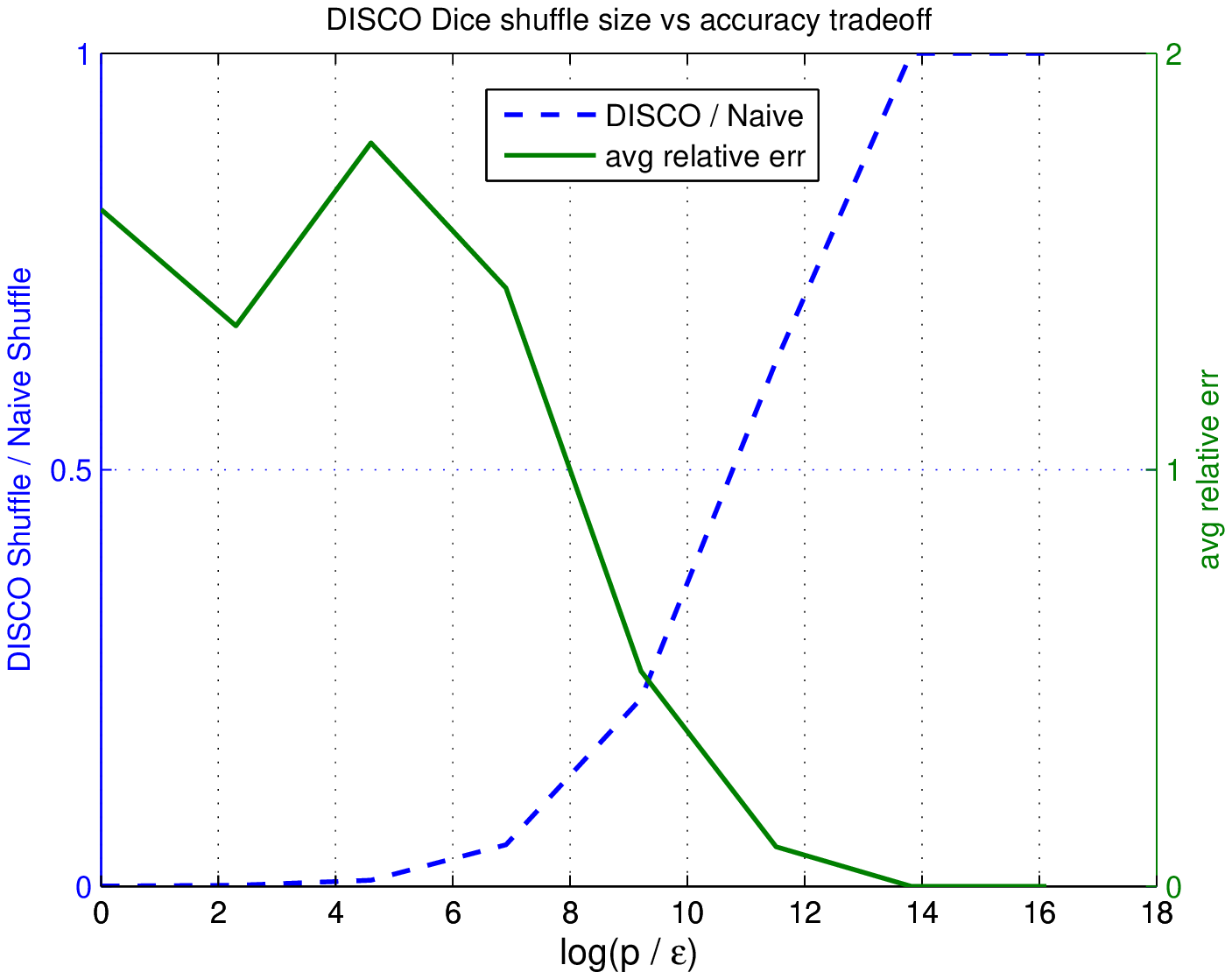}
 \caption{As $p/\epsilon$ increases, shuffle size increases and error decreases. There is no thresholding for highly similar pairs here.\label{fig:diceshuffle}}
\end{figure}

\begin{figure}[h!]
  \centering
    \includegraphics[width=0.5\textwidth]{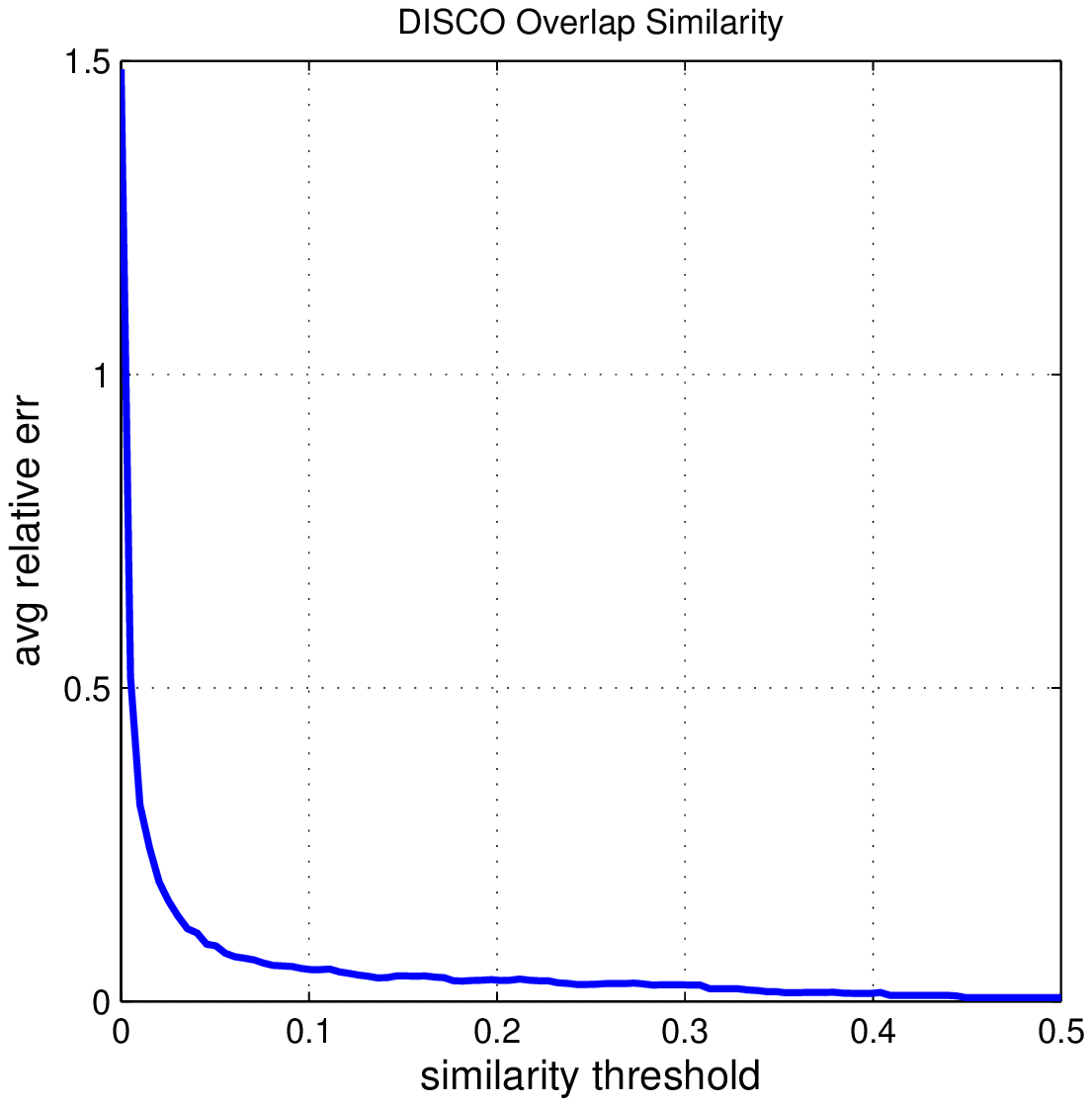}
 \caption{Average error for all pairs with similarity $\geq \epsilon$. DISCO estimated Overlap error decreases for more similar pairs. Shuffle reduction from naive implementation: 97.86\%.\label{fig:overlapallsim}}
\end{figure}
\begin{figure}[h!] 
  \centering
    \includegraphics[width=0.5\textwidth]{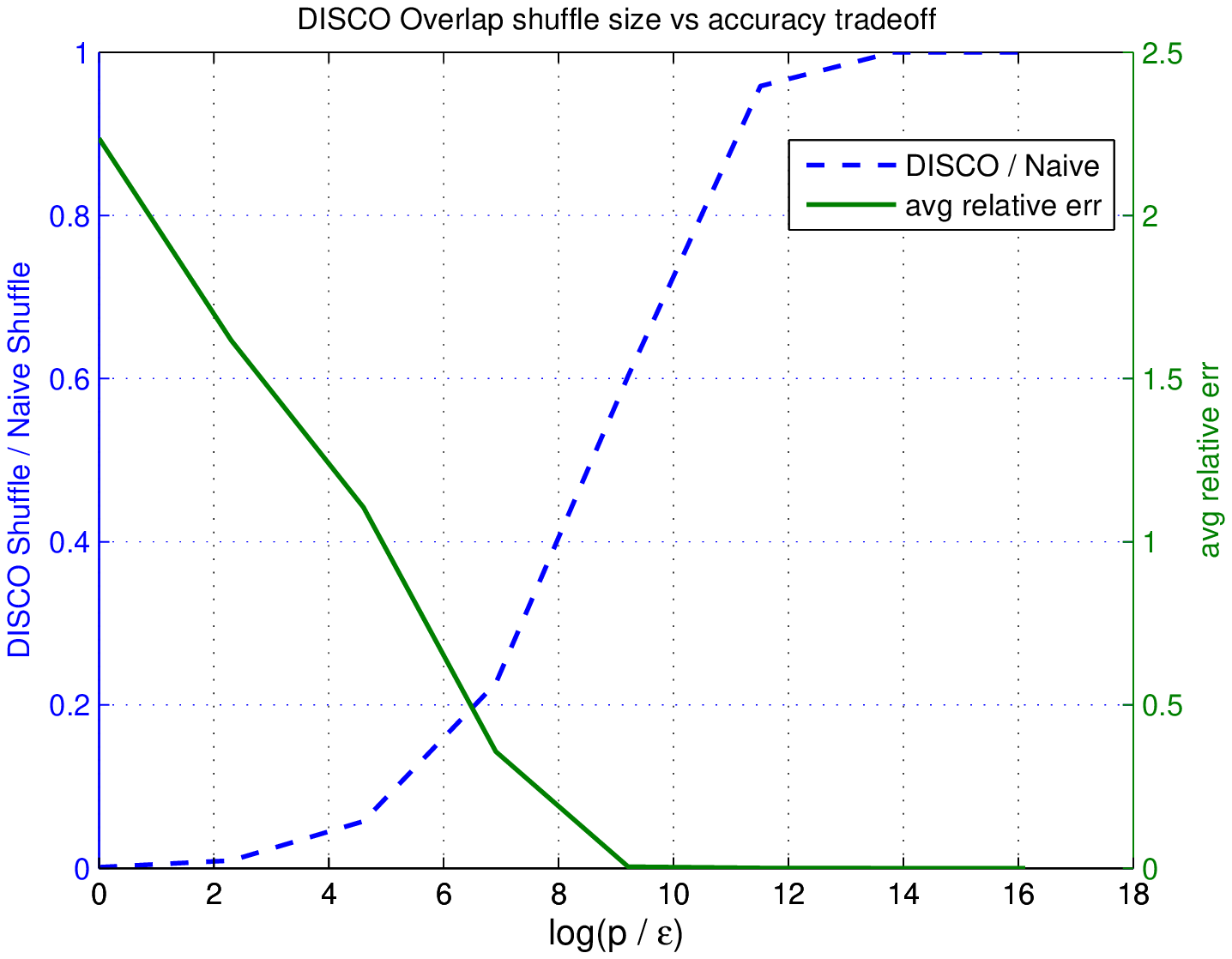}
 \caption{As $p/\epsilon$ increases, shuffle size increases and error decreases. There is no thresholding for highly similar pairs here.\label{fig:overlapshuffle}}
\end{figure}


\begin{figure}[h!]
  \centering
    \includegraphics[width=0.5\textwidth]{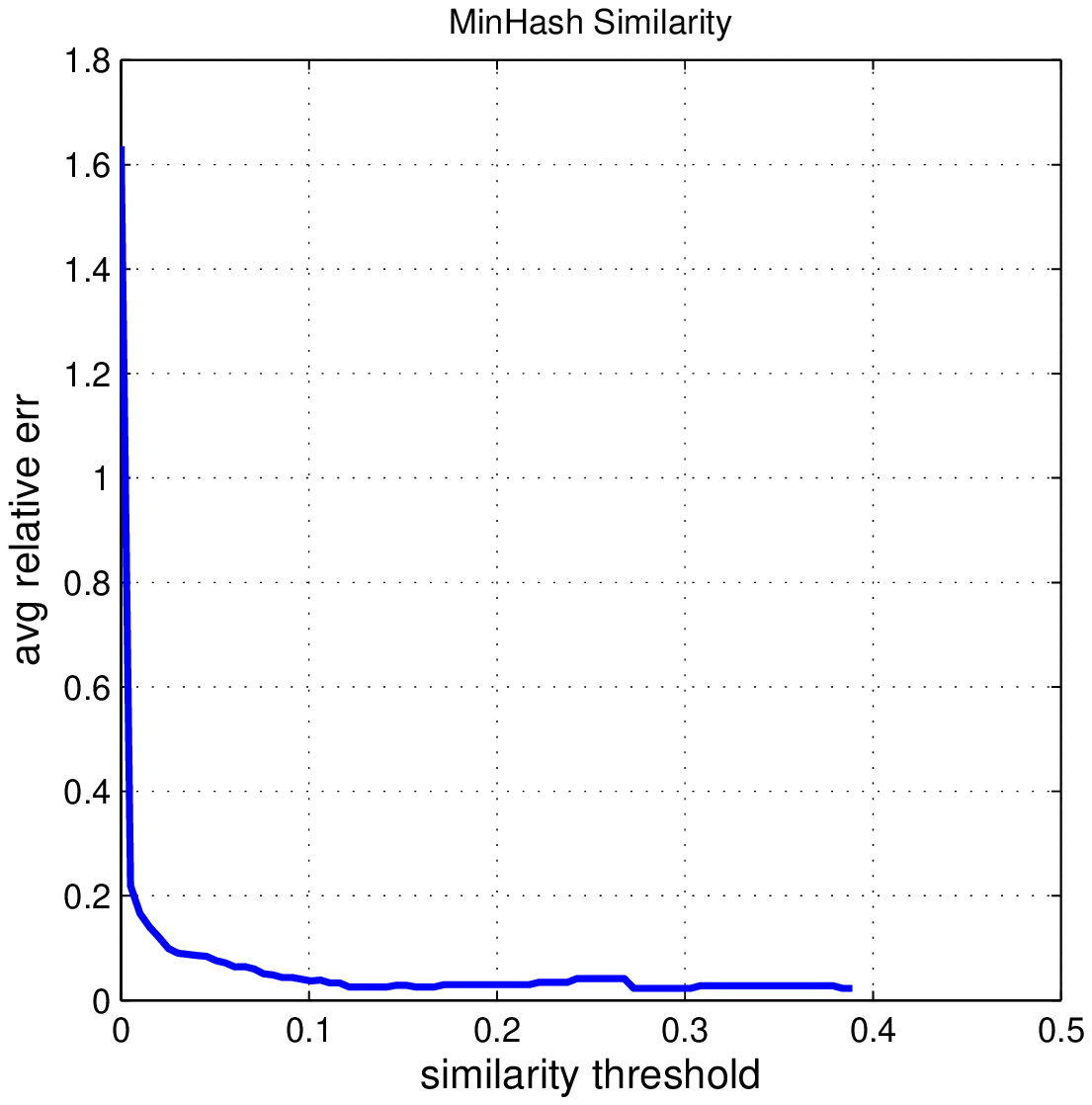}
 \caption{Average error for all pairs with similarity $\geq \epsilon$. MinHash Jaccard similarity error decreases for more similar pairs.
 We are computing error with MinHash here, not ground truth. \label{fig:mhallsim}}
\end{figure}
\newpage

\begin{figure}[h!]
  \centering
    \includegraphics[width=0.5\textwidth]{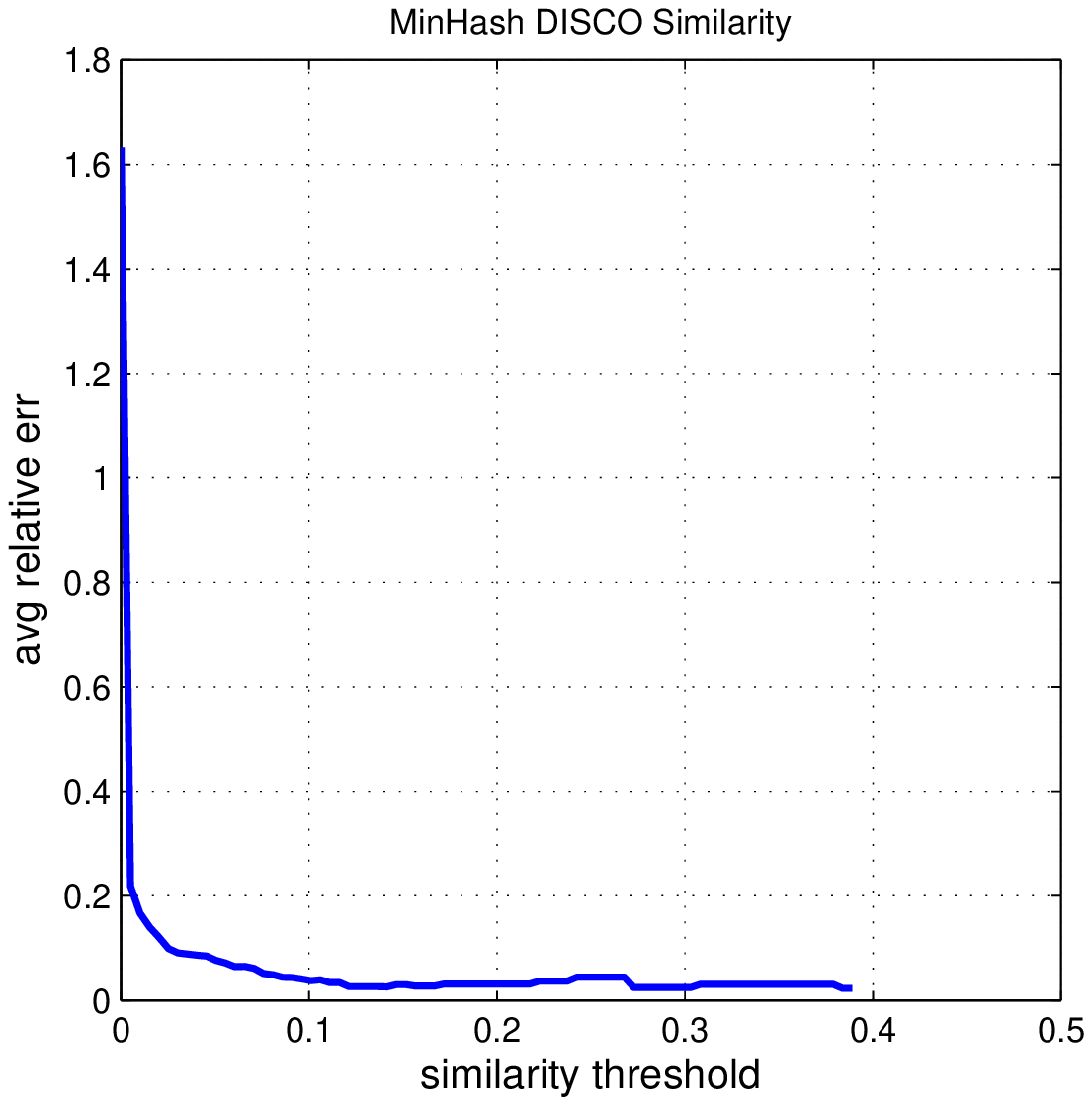}
 \caption{Average error for all pairs with similarity $\geq \epsilon$. DISCO MinHash Jaccard similarity error decreases for more similar pairs.  We are computing error with MinHash here, not ground truth. \label{fig:mhestallsim}}
\end{figure}

\subsection{Error vs Similarity Magnitude}  \label{sec:errorvsepsilon}
All of our theorems report better accuracy for pairs that have higher similarity than otherwise. To see this empirically, we 
plot the average error of all pairs that have true similarity above $\epsilon$. These can be seen in Figures \ref{fig:cosallsim}, \ref{fig:diceallsim}, \ref{fig:overlapallsim}, \ref{fig:mhallsim}, and \ref{fig:mhestallsim}. Note that the reason for 
large portions of the error being constant in these plots is that there are very few pairs with very high similarities, and therefore the
error remains constant while $\epsilon$ is between the difference of two such very high similarity pairs. Further note that
Figure \ref{fig:mhallsim} and \ref{fig:mhestallsim} are so similar because our proposed DISCO MinHash very closely
mimics the original MinHash. This is reflected in the strength of the bound in Theorem \ref{jacworks}.

\section{Conclusions and Future Directions} \label{sec:conclusions}

We presented the DISCO suite of algorithms to compute
all pairwise similarities between very high dimensional sparse vectors.
All of our results are provably independent of dimension, meaning
apart from the initial cost of trivially reading in the data, all subsequent
operations are independent of the dimension, thus the dimension can be very large.

Although we use the MapReduce \citep{mapreduce} and Streaming computation models
to discuss shuffle size and memory, the sampling strategy we use can be
generalized to other frameworks. 
We anticipate the DISCO sampling strategy to be useful whenever one is computing a number between 0 and 1
by taking the ratio of an unknown number (the dot product in our case) by some known number (e.g. $\sqrt{\#(x)\#(y)}$ for cosine similarity) \cite{dimsum}. This is a high-level description, and in order to make it more practical,
 we give five concrete examples, along with proofs, and experiments.

\section{Acknowledgements}

We thank the reviewers for their immensely helpful and thorough comments. We also thank the Twitter Personalization and Recommender Systems team for allowing use of production data. Supported in part by the DARPA xdata program, by grant FA9550-
12-1-0411 from the U.S. Air Force Office of Scientific Research
(AFOSR) and the Defense Advanced Research Projects Agency (DARPA), and
by NSF Award 0915040.

\newpage
\bibliography{disco}

\end{document}